\documentclass[acmsmall,authorversion]{ec20acm}
\copyrightyear{2020}
\acmYear{2020}
\setcopyright{acmcopyright}\acmConference[EC '20]{Proceedings of the 21st ACM Conference on Economics and Computation}{July 13--17, 2020}{Virtual Event, Hungary}
\acmBooktitle{Proceedings of the 21st ACM Conference on Economics and Computation (EC '20), July 13--17, 2020, Virtual Event, Hungary}
\acmPrice{15.00}
\acmDOI{10.1145/3391403.3399448}
\acmISBN{978-1-4503-7975-5/20/07}
\usepackage[utf8]{inputenc}

\usepackage{mathtools, bbold}
\usepackage{amsmath,amssymb,amsfonts}
\theoremstyle{plain}
\newtheorem{theorem}{Theorem}
\newtheorem{lemma}{Lemma}
\numberwithin{lemma}{subsection}
\theoremstyle{definition}
\newtheorem{definition}{Definition}

\usepackage{tikz}
\usetikzlibrary{patterns}

\newcommand{\W}{\mathcal W}
\newcommand{\M}{\mathcal M}
\newcommand{\E}{\mathbb E}
\renewcommand{\P}{\mathbb P}

\title{Two-Sided Random Matching Markets: Ex-Ante Equivalence of the Deferred Acceptance Procedures}

\author{Simon Mauras}
\email{simon.mauras@irif.fr}
\affiliation{Université de Paris, IRIF, CNRS, F-75013 Paris, France}

\authornote{This work was partially funded by the grant ANR-19-CE48-0016 from the French National Research Agency (ANR); and was done in part while the author was visiting the Simons Institute for the Theory of Computing. The author would like to gratefully thank Hugo Gimbert and Claire Mathieu for their guidance; Nicole Immorlica, Jacob Leshno, Olivier Tercieux and participants of the market design program at Simons institute for helpful and interesting discussions; and the anonymous referees for their useful comments and suggestions.}


\begin{abstract}
Stable matching in a community consisting of $N$ men and $N$ women is a classical combinatorial problem that has been the subject of intense theoretical and empirical study since its introduction in 1962 in a seminal paper by Gale and Shapley.

When the input preference profile is generated from a distribution, we study the output distribution of two stable matching procedures: women-proposing-deferred-acceptance and men-proposing-deferred-acceptance.
We show that the two procedures are ex-ante equivalent: that is, under certain conditions on the input distribution, their output distributions are identical.

In terms of technical contributions, we generalize (to the non-uniform case) an integral formula, due to Knuth and Pittel, which gives the probability that a fixed matching is stable. Using an inclusion-exclusion principle on the set of rotations, we give a new formula which gives the probability that a fixed matching is the women/men-optimal stable matching. We show that those two probabilities are equal with an integration by substitution.
\end{abstract}

\begin{document}

\maketitle
\fancyfoot[C]{\thepage}
\newpage

\section{Introduction}

Stable matching is a classical combinatorial problem, where $N$ women and $N$ men, all heterosexual and monogamous, have ordinal preferences over the persons of the opposite sex. The objective is to find a matching without any blocking pair: a woman and a man who are not married to each other but prefer each other to their actual mates.

In their seminal paper, Gale and Shapley \cite{gale1962college} proved that there always exists a stable matching, and gave a deferred acceptance procedure to find it: one side proposes, while the other side disposes. However, the \emph{men proposing deferred acceptance} (MPDA) and \emph{women proposing deferred acceptance} (WPDA) might not find the same stable matching. Moreover, Gale and Shapley showed that MPDA finds a stable matching which is optimal for the men and pessimal for the women (men have their best possible stable wife and women have their worst possible stable husband). By symmetry, WPDA finds a stable matching which is optimal for the women and pessimal for the men. Thus, there is a unique stable matching if and only if MPDA and WPDA output the same matching.

More recently, the research community in economics and computation has studied the extend to which the output of MPDA and WPDA differ in real life instances\footnote{Deferred acceptance procedures have been successfully implemented in many matching markets; see \cite{roth1999redesign, abdulkadirouglu2005new, abdulkadirouglu2005boston,correa2019school}.}, using either empirical data or stochastic models \cite{immorlica2015incentives, kojima2009incentives, ashlagi2017unbalanced, hassidim2018need,2019stablepairs}, showing that most of the time a stable matching is essentially unique (phenomenon often referred as ``core-convergence'').

Following this direction of enquiry, we consider a model where the input preference profile is generated from a distribution. We show that under certain types of input distribution, the output distributions of MPDA and WPDA are identical. This result is unusual for several reasons:

\begin{itemize}
    \item Most of the literature focuses on ``large markets'' where the core converges when the number of agents grow to infinity. Our result is stronger: output distributions are identical, and this equality holds no matter what the size of the market is.
    \item But in return, we only proved a weaker property of ``ex-ante core-convergence''; whereas most papers study ex-post the difference between outputs of MPDA and WPDA.
\end{itemize}

Our contribution in this paper is threefold:
we discovered numerically an intriguing mathematical property 
(ex-ante equivalence of WPDA and MPDA) which holds in
random matching markets with a vanilla model of preference distributions
(see Section~\ref{sec:example} and Theorem~\ref{thm:incomplete});
we identified a larger class of preference distributions for which
this property remains valid
(see Section~\ref{sec:input} and Theorem~\ref{thm:antipop});
and we formally proved this property (see Section~\ref{sec:theorem}).
Previous results that are used in our analysis are summarized in Section
\ref{sec:previous}.


\vspace{-.2cm}
\section{Motivating special case}
\label{sec:example}

Consider a random two-sided matching market, where a (given) procedure computes a stable matching.
Every agent of the market is interested by the distribution of outcomes.
But computing which outcome an agent can expect is a difficult question,
that has only been answered in special cases (for example, see \cite{lee2016incentive} for a model with random vertical preferences).
The starting point of this work was to understand the output distributions of MPDA and WPDA, in a very simple matching market with $M$ men and $W$ women having heterogeneous preferences (agents have idiosyncratic preferences).

\begin{definition}[Incomplete uniform preference distribution]
\label{def:incomplete}
  Consider any fixed bipartite graph $G = (\M \cup \W, E)$ with $\M = \{m_1, \dots, m_M\}$ the set of men, $\W = \{w_1, \dots, w_W\}$ the set of women, and $E \subseteq \M \times \W$ the set of edges. Each agent ranks his/her neighbours (non-edges are not acceptable), uniformly and independently at random. We call such input model an \emph{incomplete uniform preference distribution}.
\end{definition}

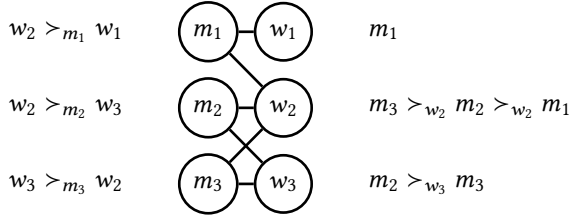
\begin{figure}[h]
    \centering
    \begin{tikzpicture}[line width=1pt]
        \node[draw,circle] (M1) at (0,2) {$m_1$};
        \node[draw,circle] (M2) at (0,1) {$m_2$};
        \node[draw,circle] (M3) at (0,0) {$m_3$};
        \node[draw,circle] (W1) at (1,2) {$w_1$};
        \node[draw,circle] (W2) at (1,1) {$w_2$};
        \node[draw,circle] (W3) at (1,0) {$w_3$};
        \draw (W1) -- (M1);
        \draw (W2) -- (M1);
        \draw (W2) -- (M2);
        \draw (W2) -- (M3);
        \draw (W3) -- (M2);
        \draw (W3) -- (M3);
        \node[anchor=west] at (2,2) {$m_1$};
        \node[anchor=west] at (2,1) {$m_3\succ_{w_2}m_2\succ_{w_2}m_1$};
        \node[anchor=west] at (2,0) {$m_2\succ_{w_3}m_3$};
        \node[anchor=east] at (-1, 2) {$w_2\succ_{m_1}w_1$};
        \node[anchor=east] at (-1, 1) {$w_2\succ_{m_2}w_3$};
        \node[anchor=east] at (-1, 0) {$w_3\succ_{m_3}w_2$};
    \end{tikzpicture}
    \caption{Example of incomplete uniform preference distribution. The probability of sampling this particular preference profile is $1 / ( 2! \cdot 2! \cdot 2! \cdot 1! \cdot 3! \cdot 2!) = 1/96$.
    There are two stable matchings,
    MPDA outputs $\{(m_1,w_1),(m_2,w_2),(m_3,w_3)\}$ and
    WPDA outputs $\{(m_1,w_1),(m_2,w_3),(m_3,w_2)\}$.}
    \label{fig:inputincomplete}
\end{figure}

Figure~\ref{fig:inputincomplete} illustrates Definition~\ref{def:incomplete}, on a bipartite graph with $3$ women and $3$ men. The output distributions of procedures MPDA and WPDA can (painfully) be computed by hand, and they happen to be identical. They are given in Figure~\ref{fig:outputincomplete}.
In particular $\P[\text{MPDA outputs }\mu_1] = \P[\text{WPDA outputs }\mu_1]$ is already a non-trivial result, as Figure~\ref{fig:inputincomplete} describes an instance where MPDA outputs $\mu_1$ and WPDA outputs $\mu_2$.

\begin{figure}[h]
    \centering
    \begin{tabular}{rcccc} Matching &
    \begin{tikzpicture}[line width=1pt, align=center]
        \node[draw,circle] (M1) at (0,2) {$m_1$};
        \node[draw,circle] (M2) at (0,1) {$m_2$};
        \node[draw,circle] (M3) at (0,0) {$m_3$};
        \node[draw,circle] (W1) at (1,2) {$w_1$};
        \node[draw,circle] (W2) at (1,1) {$w_2$};
        \node[draw,circle] (W3) at (1,0) {$w_3$};
        \draw (W1) -- (M1);
        \draw[black!10!white] (W2) -- (M1);
        \draw (W2) -- (M2);
        \draw[black!10!white] (W2) -- (M3);
        \draw[black!10!white] (W3) -- (M2);
        \draw (W3) -- (M3);
    \end{tikzpicture} &
    \begin{tikzpicture}[line width=1pt]
        \node[draw,circle] (M1) at (0,2) {$m_1$};
        \node[draw,circle] (M2) at (0,1) {$m_2$};
        \node[draw,circle] (M3) at (0,0) {$m_3$};
        \node[draw,circle] (W1) at (1,2) {$w_1$};
        \node[draw,circle] (W2) at (1,1) {$w_2$};
        \node[draw,circle] (W3) at (1,0) {$w_3$};
        \draw (W1) -- (M1);
        \draw[black!10!white] (W2) -- (M1);
        \draw[black!10!white] (W2) -- (M2);
        \draw (W2) -- (M3);
        \draw (W3) -- (M2);
        \draw[black!10!white] (W3) -- (M3);
    \end{tikzpicture}&
    \begin{tikzpicture}[line width=1pt]
        \node[draw,circle] (M1) at (0,2) {$m_1$};
        \node[draw,circle] (M2) at (0,1) {$m_2$};
        \node[draw,circle] (M3) at (0,0) {$m_3$};
        \node[draw,circle] (W1) at (1,2) {$w_1$};
        \node[draw,circle] (W2) at (1,1) {$w_2$};
        \node[draw,circle] (W3) at (1,0) {$w_3$};
        \draw[black!10!white] (W1) -- (M1);
        \draw (W2) -- (M1);
        \draw[black!10!white] (W2) -- (M2);
        \draw[black!10!white] (W2) -- (M3);
        \draw (W3) -- (M2);
        \draw[black!10!white] (W3) -- (M3);
    \end{tikzpicture}&
    \begin{tikzpicture}[line width=1pt]
        \node[draw,circle] (M1) at (0,2) {$m_1$};
        \node[draw,circle] (M2) at (0,1) {$m_2$};
        \node[draw,circle] (M3) at (0,0) {$m_3$};
        \node[draw,circle] (W1) at (1,2) {$w_1$};
        \node[draw,circle] (W2) at (1,1) {$w_2$};
        \node[draw,circle] (W3) at (1,0) {$w_3$};
        \draw[black!10!white] (W1) -- (M1);
        \draw (W2) -- (M1);
        \draw[black!10!white] (W2) -- (M2);
        \draw[black!10!white] (W2) -- (M3);
        \draw[black!10!white] (W3) -- (M2);
        \draw (W3) -- (M3);
    \end{tikzpicture}\\
     & $\mu_1$ & $\mu_2$ & $\mu_3$ & $\mu_4$ \\\\
    Probability & 19/48 & 19/48 & 5/48 & 5/48
    \end{tabular}
    \caption{Output distribution, common to the procedures MPDA and WPDA, using the input model of Figure 1.}
    \label{fig:outputincomplete}
\end{figure}
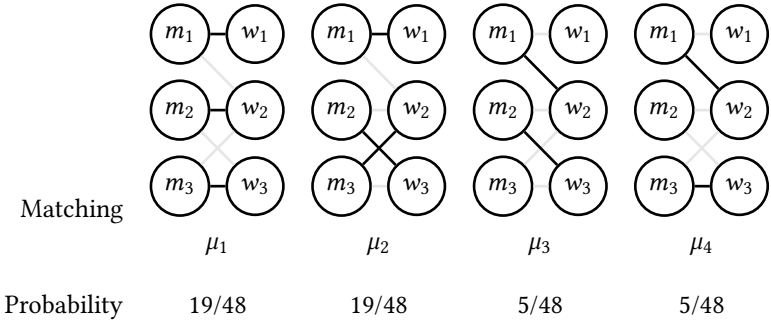

For every bipartite graph with $M,W \leq 4$, we used computer simulations to compute the output distribution of MPDA and WPDA. Surprisingly the two output distributions were always identical, which led us to conjecture Theorem \ref{thm:incomplete}.

\begin{theorem}
\label{thm:incomplete}
    In a random matching market where the preference profile is sampled from an incomplete uniform preference distribution, the output distributions of MPDA and WPDA are identical.
\end{theorem}
\begin{proof}[Proof sketch]
In Section~\ref{sec:theorem} we prove Theorem~\ref{thm:antipop}. In Subsection~\ref{sec:inputincomplete} we prove that Theorem~\ref{thm:antipop} implies Theorem~\ref{thm:incomplete}. Nonetheless, let us give the ideas of the proof on the example of Figures~\ref{fig:inputincomplete} and~\ref{fig:outputincomplete}.

\medskip
To compute the probability that WPDA outputs $\mu_1$, we first compute the probability that $\mu_1$ is stable. Then we subtract the probability that $\mu_1$ is stable but not women-optimal, because $\mu_2$ is stable and improves the outcome of women. It turns out that the probability that $\mu_2$ improves the outcome of women is equal to the probability that $\mu_2$ improves the outcome of men, hence MPDA and WPDA have the same probability of outputting $\mu_1$.

\medskip
In more complicated instances, where there are several ways to improve the outcome of women, we use an inclusion-exclusion principle on the set of rotations (defined in Subsection~\ref{sec:previouslattice}).
\end{proof}

\section{Previous Results}
\label{sec:previous}

In this section, we summarize previous results that are used in our analysis, and define the notations that will be used in the rest of the paper. First, we argue that without loss of generality we can consider matching markets with an equal number of men and women, and where every couple is acceptable (this will be discussed further in Subsection~\ref{sec:inputincomplete}). Then, we review classical results on the structure of the set of stable matchings. Finally, we recall existing formulas which give the probability of stability of a fixed matching.

\subsection{Stable Matchings}
\label{sec:previousmatching}

Stable matchings were introduced in 1962 by Gale and Shapley \cite{gale1962college}.
Let us start with formal definitions and classical notations. Let $N$ be an integer, $\M = \{m_1, \dots, m_N\}$ be a set of men, and let let $\W = \{w_1, \dots, w_N\}$ be a set of women. Each man $m$ has a total ordering $(\succ_m)$ over the women, and each woman $w$ has a total ordering $(\succ_w)$ over the men.
We view a \emph{matching} as a function $\mu : \M \cup \W \rightarrow \M \cup \W$, which is an involution ($\mu^2 = \text{Id}$), where each man is paired with a woman ($\mu(\M) \subseteq \W$), and each woman is paired with a man ($\mu(\W) \subseteq \M$). A matching is \emph{stable} if there are no blocking pairs $(m,w)$, where $m$ prefers $w$ to his wife and $w$ prefers $m$ to her husband.
\[\text{Matching }\mu\text{ is stable} \quad\Leftrightarrow\quad \forall m \in \M, \forall w \in \W, (\mu(m) \succeq_m w)  \text{ or } (\mu(w) \succeq_w m)\]

In the classical definition of stable matchings, the preference lists are complete. However, in most applications, it is not practical to ask every agent to report a full preference list. In a generalization of stable matchings, people can declare some members of the opposite sex to be unacceptable. In this setting, stable matchings may not be perfect matchings, and a pair $(m,w)$ can block a matching only if both $m$ and $w$ declare each other to be acceptable.
A further generalization allows the set of men and the set of women to be of different size.

\medskip
In Section~\ref{sec:example}, we described a model of stable matchings with unacceptable partners. Without loss of generality, it is enough to study balanced matching markets with complete preference list. Firstly, dealing with unbalanced matching markets is easy: it is always possible to add ``virtual'' persons that are unacceptable to everyone from the opposite sex. Secondly, in matching markets with unacceptable partners, the set of people that are matched is the same in every stable matching. Therefore, running MPDA (resp. WPDA) with unacceptable partners is equivalent with the following procedure: first we symmetrize unacceptability (such that $w$ is acceptable to $m$ if and only if $m$ is acceptable to $w$), second we append unacceptable partners at the end of preference lists (in any arbitrary order), third we run MPDA (resp. WPDA) on this new instance, fourth we remove couples that were not acceptable.

\subsection{Lattice of stable matchings}
\label{sec:previouslattice}

Given in input the preference profile (containing all the preference lists), the procedure MPDA outputs a stable matching $\mu_\M$ which is optimal for the men and pessimal for the women. Symmetrically, WPDA outputs a stable matching $\mu_\W$ which is optimal for the women and pessimal for the men.
\begin{align*}
    \forall \mu\text{ stable matching},\quad
    &\forall m \in \M, \quad \mu_\M(m) \succeq_m \mu(m) \succeq_m \mu_\W(m)\\
    &\forall w \in \W, \quad \mu_\W(w) \succeq_w \mu(w) \succeq_w \mu_\M(w)
\end{align*}
The structure of the set of stable matchings was studied by Knuth and Conway
\cite{knuth1976mariages,knuth1997stable}: with the partial orders $\succeq_\M$ and $\succeq_\W$, the set of stable matching is a distributive lattice.
\begin{align*}
    \forall \mu_1,\mu_2 \text{ stable matchings},\quad
    & \mu_1 \succeq_\M \mu_2 \;\Leftrightarrow\; \forall m\in\M,\mu_1(m) \succeq_m \mu_2(m)\\ 
    & \mu_1 \succeq_\W \mu_2 \;\Leftrightarrow\; \forall w\in\W,\mu_1(w) \succeq_w \mu_2(w)\\
    & \mu_1 \succeq_\W \mu_2 \;\Leftrightarrow\; \mu_2 \succeq_\M \mu_1
\end{align*}

The concept of rotation was later introduced by Irving, Leather and Gusfield (see \cite{gusfield1989stable} for a nice survey). We view a \emph{rotation} as a simple directed cycle $r$ in the complete bipartite graph over $\M \cup \W$. When a person $x \in \M\cup\W$ belongs to the cycle, we write $x \in r$, denote $r(x)$ a successor and $r^{-1}(x)$ a predecessor.
In a stable matching $\mu_1$, rotation $r$ is exposed and women-improving if for all man $m$, $r(m) = \mu_1(m)$ and $r^{-1}(m)$ is $m$'s favourite woman among women $w$ to whom he prefers his wife ($\mu_1(m) \succ_m w$), and who prefer $m$ to their husband ($m \succ_w \mu_1(w)$).
Eliminating rotation $r$ in matching $\mu_1$ creates a new stable matching $\mu_2$; we have $\mu_2 \succeq_\W \mu_1$.
\[\forall m\in M,\; \mu_2(m) = \left\{\begin{array}{cc}
    r^{-1}(m) & \text{if }m \in r \\
    \mu_1(m) & \text{if }m \notin r
\end{array}\right.\qquad
\forall w\in\W,\; \mu_2(w) = \left\{\begin{array}{cc}
    r(w) & \text{if }w \in r \\
    \mu_1(w) & \text{if }w \notin r
\end{array}\right.\]
Symmetrically, rotation $r$ is exposed and men-improving in stable matching $\mu_2$. Eliminating $r$ in $\mu_2$ creates stable matching $\mu_1$; we have $\mu_1 \succeq_\M \mu_2$.

\begin{figure}[h]
    \centering
    \begin{tabular}{rcl}
    \begin{tikzpicture}[line width=1pt]
        \node[draw,circle] (M1) at (0,2) {$m_1$};
        \node[draw,circle] (M2) at (0,1) {$m_2$};
        \node[draw,circle] (M3) at (0,0) {$m_3$};
        \node[draw,circle] (W1) at (1,2) {$w_1$};
        \node[draw,circle] (W2) at (1,1) {$w_2$};
        \node[draw,circle] (W3) at (1,0) {$w_3$};
        \draw (W1) -- (M1);
        \draw (W2) -- (M2);
        \draw (W3) -- (M3);
        \node[anchor=east] at (-.5, 2) {$w_2\succ_{m_1}w_1\succ_{m_1}w_3$};
        \node[anchor=east] at (-.5, 1) {$w_2\succ_{m_2}w_3\succ_{m_2}w_1$};
        \node[anchor=east] at (-.5, 0) {$w_3\succ_{m_3}w_2\succ_{m_3}w_1$};
    \end{tikzpicture} &
    \begin{tikzpicture}[line width=1pt]
        \node[draw,circle] (M1) at (0,2) {$m_1$};
        \node[draw,circle] (M2) at (0,1) {$m_2$};
        \node[draw,circle] (M3) at (0,0) {$m_3$};
        \node[draw,circle] (W1) at (1,2) {$w_1$};
        \node[draw,circle] (W2) at (1,1) {$w_2$};
        \node[draw,circle] (W3) at (1,0) {$w_3$};
        \draw[->] (M2) -- (W2);
        \draw[->] (W3) -- (M2);
        \draw[->] (M3) -- (W3);
        \draw[->] (W2) -- (M3);
    \end{tikzpicture} &
    \begin{tikzpicture}[line width=1pt]
        \node[draw,circle] (M1) at (0,2) {$m_1$};
        \node[draw,circle] (M2) at (0,1) {$m_2$};
        \node[draw,circle] (M3) at (0,0) {$m_3$};
        \node[draw,circle] (W1) at (1,2) {$w_1$};
        \node[draw,circle] (W2) at (1,1) {$w_2$};
        \node[draw,circle] (W3) at (1,0) {$w_3$};
        \draw (W1) -- (M1);
        \draw (W2) -- (M3);
        \draw (W3) -- (M2);
        \node[anchor=west] at (1.5,2) {$m_1\succ_{w_1}m_2\succ_{w_1}m_3$};
        \node[anchor=west] at (1.5,1) {$m_3\succ_{w_2}m_2\succ_{w_2}m_1$};
        \node[anchor=west] at (1.5,0) {$m_2\succ_{w_3}m_3\succ_{w_3}m_1$};
    \end{tikzpicture} \\
    Matching $\mu_1$ & Rotation $r$ & Matching $\mu_2$
    \end{tabular}
    \caption{Example of rotation $r$, women-improving from $\mu_1$ to $\mu_2$, men-improving from $\mu_2$ to $\mu_1$.}
    \label{fig:rotation}
\end{figure}
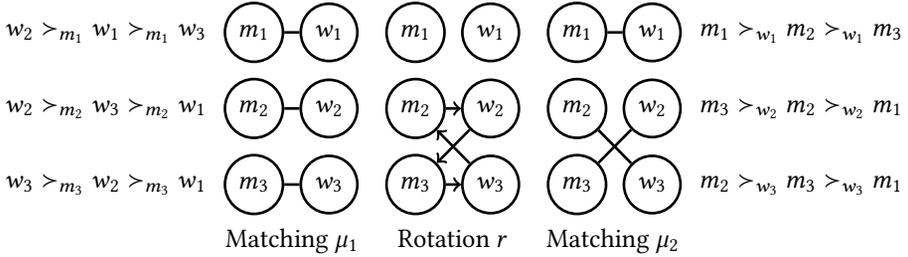

If several rotations are exposed and women-improving in the same stable matching, those rotations are disjoint. To describe several disjoint rotations exposed in the same matching, we will use the concept of stable permutations. Stable permutations (also called stable partitions) have been defined for the more general problem of stable roommates \cite{tan1991necessary}. A \emph{permutation} is a bijection $\sigma : \M\cup\W \rightarrow \M\cup\W$, where the successor of a man is a woman ($\sigma(\M) \subseteq \W$), and where the successor of a woman is a man ($\sigma(\W) \subseteq \M$). A permutation $\sigma$ is \emph{stable} if each person $x$ prefers their successor to their predecessor ($\sigma(x) \succeq_x \sigma^{-1}(x)$), and if there are no blocking pairs $(m,w)$, where $m$ prefers $w$ to his predecessor ($w \succ_m \sigma^{-1}(m)$) and $w$ prefers $m$ to her predecessor ($m \succ_w \sigma^{-1}(w)$).

To conclude this subsection on the structure of stable matchings, we summarize properties on matchings, rotations, and permutations. Those properties will be useful in Section~\ref{sec:theorem}.
\begin{itemize}
    \item A matching $\mu$ is stable (as a matching) if and only if it is stable (as a permutation).
    \item Let $\sigma$ be a permutation. Every cycle of length $>2$ of $\sigma$ is a rotation. 
    Let $\mu_1$ be the only matching such that $\mu_{1|\M} = \sigma_{|\M}$. Let $\mu_2$ be the only matching such that $\mu_{2|\W} = \sigma_{|\W}$. 
    Permutation $\sigma$ is stable if and only if matchings $\mu_1$ and $\mu_2$ are both stable, and every rotation induced by $\sigma$ is exposed and women improving (resp. men improving) in $\mu_1$ (resp. $\mu_2$).
\end{itemize}

\subsection{Probability of Stability}
\label{sec:previousproba}

Random matching markets with $N$ men and $N$ women having uniformly random preference lists were studied in \cite{knuth1976mariages, knuth1997stable, pittel1989average, pittel1992likely}.
Knuth gave an integral formula for the probability $p_N$ that a fixed matching is stable;
with the objective of computing the asymptotic average number of stable matchings
(in the uniform case, all $N!$ matchings have the same probability of being stable).
In 1989, Pittel gave an alternate proof of this integral formula, and showed that $N! \cdot p_N \sim  e^{-1} N \ln N$.

\medskip
Let us retranscribe Pittel's proof of the integral formula. Let $\mu$ be any matching. Let $X$ and $Y$ be two random matrices, uniformly sampled from $[0,1]^{\M\times\W}$. Man $m$ prefers woman $w_1$ to woman $w_2$ if $X_{m,w_1} < X_{m,w_2}$. Correspondingly, woman $w$ prefers man $m_1$ to man $m_2$ if $Y_{m_1,w} < Y_{m_2,w}$.

\medskip
Thus, a pair $(m,w)$ is blocking matching $\mu$ if and only if $X_{m,w} < X_{m,\mu(m)}$ and $Y_{m,w} < Y_{\mu(w),m}$.
We condition on the values of $\mathbf x = [X_{m,\mu(m)}]_{m\in\M}$ and $\mathbf y = [Y_{\mu(w),w}]_{w\in\W}$, and write the probability that a pair blocks $\mu$:
\[\forall (m,w) \text{ such that }\mu(m) \neq w\text{ and }\mu(w) \neq w,\quad
\P[(m,w) \text{ blocks }\mu\,|\,\mathbf x, \mathbf y] = \mathbf x_m \cdot \mathbf y_w\]
Still conditioning on $\mathbf x$ and $\mathbf y$, blocking events are independent, hence the formula:
\[\P[\mu \text{ is stable}] = \underbrace{\int\dots\int}_{2N} \mathrm d\mathbf x \cdot \mathrm d \mathbf y \cdot
\!\!\prod_{\substack{m,w\\\mu(m) \neq w\\\mu(w)\neq m}} (1 - \mathbf x_m \mathbf y_w)\]

In subsequent works \cite{pittel1994upper,pittel2019random}, Pittel extended the above formula to compute the probability that a fixed permutation is stable. We recall that a permutation $\sigma$ is stable if the following is true:
\begin{itemize}
    \item Every person $x$ prefers their successor to their predecessors ($\sigma(x) \succeq_x \sigma^{-1}(x)$)
    \item For each pair $(m,w) \in \M\times\W$, we have
    $(\sigma^{-1}(m) \succeq_m w)  \text{ or } (\sigma^{-1}(w) \succeq_w m)$
\end{itemize}
We condition on the values of $\mathbf x = [X_{m,\sigma^{-1}(m)}]_{m\in\M}$ and $\mathbf y = [Y_{\sigma^{-1}(w),w}]_{w\in\W}$.
\begin{itemize}
    \item Each man $m$ such that $\sigma(m) \neq \sigma^{-1}(m)$ prefers $\sigma(m)$ to $\sigma^{-1}(m)$ with probability $\mathbf x_m$.
    \item Each woman $w$ such that $\sigma(w) \neq \sigma^{-1}(w)$ prefers $\sigma(w)$ to $\sigma^{-1}(w)$ with probability $\mathbf y_w$.
    \item Each pair $(m,w)$ such that $\sigma(m) \neq w$ and $\sigma(w) \neq m$ is blocking with probability $\mathbf x_m \mathbf y_w$.
\end{itemize}
Hence the formula:
\[\P[\sigma \text{ is stable}] = \underbrace{\int\dots\int}_{2N} \mathrm d\mathbf x \cdot \mathrm d \mathbf y \cdot
\prod_{\substack{m,w\\\sigma(m) = w\\\sigma(w) \neq m}} \mathbf x_m \cdot
\prod_{\substack{m,w\\\sigma(m) \neq w\\\sigma(w) = m}} \mathbf y_w \cdot
\!\!\prod_{\substack{m,w\\\sigma(m) \neq w\\\sigma(w) \neq m}} (1 - \mathbf x_m \mathbf y_w)\]

For the more general problem of stable roommates, Mertens \cite{mertens2015small}
combined this formula with an inclusion-exclusion principle to compute the
probability that a random instance has a solution.

\section{Input Model}
\label{sec:input}

In this section, we describe the preference distribution that will be our input model. Subsection~\ref{sec:inputweighted} gives the most general definition, Subsections~\ref{sec:inputincomplete} and \ref{sec:inputvertical} detail interesting special cases.

\subsection{Symmetric preference distributions}
\label{sec:inputweighted}

After observing that the output distributions of MPDA and WPDA are identical when the preference profile is generated from a bipartite graph (see Section~\ref{sec:example}), we used computer simulations on more general classes of input distributions. We observed that MPDA and WPDA are ex-ante equivalent with the input model illustrated in Figure~\ref{fig:symmetricantipop} and defined in Definition~\ref{def:symmetricantipop}.

\begin{figure}[h]
    \centering
    \small
    \begin{minipage}{.4\textwidth}
    \[P \;=\quad \begin{array}{c|c|c|c|}
    & w_1 & w_2 & w_3\\\hline
    m_1 & 2 & 1 & 3\\\hline
    m_2 & 5 & 6 & 2\\\hline
    m_3 & 3 & 4 & 1\\\hline
    \end{array}\]
    \par\medskip
    The preference list of $w_1$ is $m_2 \succ m_1 \succ m_3$ with probability:
    \[\frac{1/3}{1/2+1/5+1/3} \cdot \frac{1/2}{1/2+1/5} \cdot \frac{1/5}{1/5} \approx 0.23\]
    \end{minipage}
    \quad
    \begin{minipage}{.55\textwidth}
    \scalebox{.8}{\begin{tikzpicture}[>=stealth,xscale=.85]
    \node[draw, circle] (T) at (10,0) {};
    \node[draw] (T1) at (6,1.5) {$m_1$ is 3\textsuperscript{rd}};
    \node[draw] (T2) at (6,0) {$m_2$ is 3\textsuperscript{rd}};
    \node[draw] (T3) at (6,-1.5) {$m_3$ is 3\textsuperscript{rd}};
    \node[] (T01) at (4.5,1.5) {};
    \node[] (T02) at (4.5,0) {};
    \node[draw] (T13) at (2.5,-.5) {$m_1$ is 2\textsuperscript{nd}};
    \node[draw] (T23) at (2.5,-1.5) {$m_2$ is 2\textsuperscript{nd}};
    \node[] (T023) at (1,-1.5) {};
    \node[draw] (T213) at (0,-.5) {$m_2$ is 1\textsuperscript{st}};
    \draw[->] (T) -- (T1.east);
    \draw[->] (T) -- (T2.east);
    \draw[->] (T) -- (T3.east);
    \draw[->] (T3) -- (T13.east);
    \draw[->] (T3) -- (T23.east);
    \draw[->,dashed] (T1) -- (T01.east);
    \draw[->,dashed] (T2) -- (T02.east);
    \draw[->] (T13) -- (T213.east);
    \draw[->,dashed] (T23) -- (T023.east);
    \node at (9,1.3) {$\frac{1/2}{1/5+1/3+1/2}$};
    \node at (8,0.3) {$\frac{1/5}{1/5+1/3+1/2}$};
    \node at (9,-1.1) {$\frac{1/3}{1/5+1/3+1/2}$};
    \node at (4.25,-.4) {$\frac{1/2}{1/5+1/2}$};
    \node at (4.25,-1.2) {$\frac{1/5}{1/5+1/2}$};
    \node at (1.25,-0.2) {$\frac{1/5}{1/5}$};
    \end{tikzpicture}}
    \end{minipage}
    \caption{Symmetric anti-popularity preference distribution}
    \label{fig:symmetricantipop}
\end{figure}
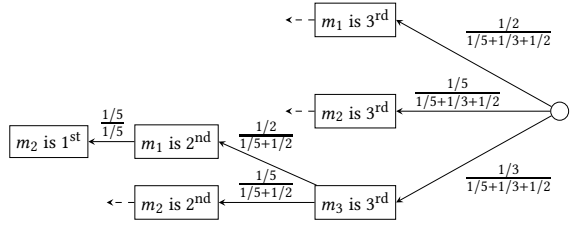

\par\bigskip
The stochastic process used to generate preference lists is very similar (but not equivalent) to the model studied in \cite{immorlica2015incentives,kojima2009incentives,2019stablepairs}. Agents build their preference lists by sampling without replacement from a distribution (in \cite{immorlica2015incentives,kojima2009incentives,2019stablepairs} agents first sample their favourite partner, in this paper agents first sample their least preferred partner).

\begin{definition}[Symmetric anti-popularity preference distribution]
\label{def:symmetricantipop}
    Consider any function $P : \M\times\W \rightarrow \mathbb R_{>0}$, where  $P(m,w)$ is the ``popularity'' that $m$ and $w$ attribute to each other. Each man $m$ first builds an ``anti-popularity'' distribution over the women, where woman $w$ has a probability of $1/P(m,w)$ (renormalized such that the sum of probability is 1); then he builds his preference list from the end, by sampling without replacement from this ``anti-popularity'' distribution: he first samples his least preferred partner, then his second least, ..., then his favorite partner. Symmetrically, each woman $w$ builds a preference list using her ``anti-popularity'' distribution over the men.
\end{definition}
We say that this preference distribution is symmetric because the ``popularity'' that $m$ gives to $w$ is the same as the ``popularity'' that $w$ gives to $m$. The ``popularity'' parameter $P(m,w)$ relates to how likely are $m$ and $w$ to like each other. In particular, a woman $w$ will prefer man $m_1$ to man $m_2$ with probability $P(m_1,w) / (P(m_1,w) + P(m_2,w))$.

\medskip
Definition~\ref{def:symmetricantipop} has several equivalent formulations. Definition~\ref{def:symmetricpower} will be used in the proof in Section~\ref{sec:theorem}. Definition~\ref{def:symmetricmemoryless} is an intermediate formulation, useful to prove the equivalence of Definitions~\ref{def:symmetricantipop} and \ref{def:symmetricpower}. For a direct proof of equivalence, see Lemma 4 of \cite{efraimidis2006weighted}.

\begin{definition}[Symmetric memoryless utility preference distribution]
\label{def:symmetricmemoryless}
Utility preferences are defined by a collection of values $(U_{m,w}, V_{m,w})$, where $U_{m,w}$ is the utility that man $m$ gets if he is matched with $w$, where $V_{m,w}$ is the utility that woman $w$ gets if she is matched with $m$, and where each agent wants to maximize their utility.
Consider only the cases where all $U_{m,w}$ and $V_{m,w}$ are independent random variables, such that the expected values are symmetric (that is $\E[U] = \E[V]$), and such that each coefficient is memoryless (that is $\P[X > s+t \;|\; X > t] = P[X > s]$ for $s,t > 0$).
\end{definition}

\begin{lemma}
\label{lemma:symmetric1}
 Symmetric anti-popularities and symmetric memoryless utilities induce the same class of preference distributions.
\end{lemma}

\begin{proof}
 Consider any symmetric anti-popularity preference distribution. For all man $m$ and woman $w$, define independent exponential random variables $U_{m,w}$ and $V_{m,w}$ of parameter $1/P(m,w)$.
 In the stochastic process using anti-popularities, someone's last choice is independent with the beginning of their preference list, which is analogous to the memorylessness property of exponential random variables. 
 Conversely, any memoryless continuous random variable in an exponential random variable. For all man $m$ and woman $w$, we define $P(m,w) = \E[U_{m,w}] = \E[V_{m,w}]$.
\end{proof}

\begin{definition}[Symmetric power preference distribution]
\label{def:symmetricpower}
Consider any function $P : \M\times\W \rightarrow \mathbb R_{>0}$, where $P(m,w)$ is the ``power'' that $m$ and $w$ attribute to each other. Let $X$ and $Y$ be two random matrices, uniformly sampled from $[0,1]^{\M\times\W}$. Values of $X$ and $Y$ induce a preference profile:
\[\forall m \in \M, \forall w_1, w_2 \in \W,\quad
w_1 \succ_m w_2 \;\Leftrightarrow\; X_{m,w_1}^{P(m,w_1)} < X_{m,w_2}^{P(m,w_2)}\]
\[\forall w \in \W, \forall m_1, m_2 \in \M,\quad
m_1 \succ_w m_2 \;\Leftrightarrow\; Y_{m_1,w}^{P(m_1,w)} < X_{m_2,w}^{P(m_2,w)}\]
\end{definition}

\begin{lemma}
\label{lemma:symmetric2}
 Symmetric powers and symmetric memoryless utilities induce the same class of preference distributions.
\end{lemma}
\begin{proof}
 Consider a power distribution defined by a function $P : \M\times\W \rightarrow \mathbb R_{>0}$ and two random matrices $X$ and $Y$, uniformly sampled from $[0,1]^{\M\times\W}$. For all man $m$ and woman $w$, let $U_{m,w} = -P(m,w) \cdot \ln X_{m,w}$ and $V_{m,w} = -P(m,w) \cdot \ln Y_{m,w}$. Then $U_{m,w}$ and $V_{m,w}$ are independent exponential random variables, which define symmetric memoryless utilities. The two preference profile distributions are identical.
Conversely, any memoryless continuous distribution is an exponential random variable.
 For all man $m$ and woman $w$, we define $P(m,w) = \E[U_{m,w}] = \E[V_{m,w}]$.
\end{proof}

\subsection{Incomplete preference distribution}
\label{sec:inputincomplete}

In this subsection, we prove that Definition~\ref{def:symmetricantipop} is strictly more general that Definition~\ref{def:incomplete}. In the sense that Theorem~\ref{thm:antipop} implies Theorem~\ref{thm:incomplete}. The main technique is to approximate an incomplete uniform preference distribution with a symmetric anti-popularity distribution, where the popularity of an edge of the bipartite graph is 1 and the popularity of a non-edge approaches 0.

\begin{proof}[Proof (Theorem~\ref{thm:antipop} $\Rightarrow$ Theorem~\ref{thm:incomplete})]
Let $(\M\cup\W,E)$ be any bipartite graph. For all $\varepsilon > 0$, we define a random preference profile $\text{Pref}_\varepsilon = (\succ^\varepsilon_\cdot)$, sampled from a symmetric anti-popularity preference distribution with a popularity function $P_\varepsilon = \varepsilon + (1 - \varepsilon) \cdot \mathbb 1_E$, where $\mathbb 1_E$ is the indicator function of the set of edges $E$.

\medskip
Let $\text{Pref}_{\varepsilon|E}$ be the same preference profile, where every couple not in the set of edges $E$ is declared as unacceptable. Observe that for every $\varepsilon > 0$, the distribution of the random preference profile $\text{Pref}_{\varepsilon|E}$ is exactly the incomplete uniform preference distribution induced by the graph $(\M\cup\W,E)$.

\medskip
Now let us describe the typical behavior of the preference profile $\text{Pref}_\varepsilon$.
We say that event $OK_\varepsilon$ holds if every edge of $E$ precedes every non-edge of $E$.
More formally, we have:
\begin{align*}
    \text{Event }OK_\varepsilon\text{ holds} \qquad\Leftrightarrow\qquad
    \forall (m,w) \in E,\qquad
    &\forall w' \in \W,\quad (m,w') \notin E \;\Rightarrow\; w \succ_m^\varepsilon w'\\
    &\forall m' \in \W,\quad (m',w) \notin E \;\Rightarrow\; m \succ_w^\varepsilon m'
\end{align*}
Using an union bound, event $OK_\varepsilon$ holds with probability at least $1 - \varepsilon N^3$.

\medskip
We see MPDA and WPDA as functions, which take as input a preference profile and output a matching (seen as a subset of $\M\times\W$). Using Subsection~\ref{sec:previousmatching}, we know that when $OK_\varepsilon$ holds, preference profiles $\text{Pref}_\varepsilon$ and $\text{Pref}_{\varepsilon|E}$ have the same stable matchings (when restricted to $E$).
\begin{align*}
    \text{Event }OK_\varepsilon\text{ holds} \qquad\Rightarrow\qquad
    & \text{MPDA}(\text{Pref}_\varepsilon) \cap E = \text{MPDA}(\text{Pref}_{\varepsilon|E})\\
    & \text{WPDA}(\text{Pref}_\varepsilon) \cap E = \text{WPDA}(\text{Pref}_{\varepsilon|E})
\end{align*}
Theorem~\ref{thm:antipop} says that for every $\varepsilon > 0$ the random variables $\text{MPDA}(\text{Pref}_\varepsilon)$ and $\text{WPDA}(\text{Pref}_\varepsilon)$ have the same distribution. This equality, as $\varepsilon$ approaches $0$, proves Theorem~\ref{thm:incomplete}.
\end{proof}

\subsection{Vertical preference distribution}
\label{sec:inputvertical}

Symmetric anti-popularities are useful to model ``cross sided'' preferences. As an example, consider the market of PhDs and post-doc positions. Imagine that Alice has a very good thesis in Computer Science, because of her skills she will most likely apply for a post-doc in a very good computer science department; symmetrically this university will most likely rank Alice first.

However, one might be afraid that symmetric anti-popularities do not encompass ``one sided'' preferences: all the PhDs might prefer university X to university Y, and all the universities might prefer Alice to Bob. The input distributions studied in \cite{immorlica2015incentives,kojima2009incentives,2019stablepairs} were able to model this kind of examples.
In this subsection we answer this concern, proving that \emph{symmetric anti-popularities} are strictly more general than \emph{vertical anti-popularities}.

\begin{definition}[Vertical anti-popularity preference distribution]
Consider two functions $P_\M : \M \rightarrow \mathbb R_{>0}$ and $P_\W : \W \rightarrow \mathbb R_{>0}$, where $P_\M(m)$ is the popularity that all the women give to man $m$, and where $P_\W(w)$ is the popularity that all the men give to woman $w$. Men first build an ``anti-popularity'' distribution over the women, where woman $w$ has a probability of $1/P_\W(w)$ (renormalized such that the sum of probability is 1); then each man builds his preference list from the end, by sampling without replacement from this ``anti-popularity'' distribution: he first samples his least preferred partner, then his second least, ..., then his favorite partner. Symmetrically, each woman $w$ builds a preference list using the ``anti-popularity'' distribution over the men.
\end{definition}

\begin{lemma}
  Vertical anti-popularities can be simulated with symmetric anti-popularities.
\end{lemma}

\begin{proof}
  For all man $m$ and woman $w$, define $P(m,w) = P_\M(m) \cdot P_\W(w)$. Because of the renormalization step, the symmetric anti-popularity preference distribution is identical to the vertical anti-popularity preference distribution.
\end{proof}

\section{Main Theorem}
\label{sec:theorem}

In this section, we prove the main result of this paper. The proof of Theorem~\ref{thm:antipop} is split in three steps, organized in three subsections. First, we generalize to our input model the formula which gives the probability that a permutation is stable. Second we prove that a permutation and its inverse are equally likely to be stable. Third, we compute the output distributions of MPDA and WPDA using the probability of stability of permutations. 

\begin{theorem}
\label{thm:antipop}
    In a random matching market where the preference profile is sampled from a symmetric anti-popularity preference distribution, the output distributions of MPDA and WPDA are identical.
\end{theorem}

\begin{proof}
 A fixed matching $\mu$ is outputted by WPDA if and only if it is stable and women-optimal. In  Lemma~\ref{lemma:optimal}, we give a formula for the probability that $\mu$ is stable and women-optimal.
\[\P[\mu\text{ is stable and women-optimal}] = \sum_{\substack{\sigma\text{ permutation}\\\sigma_{|\M} = \mu_{|\M}}} (-1)^{C(\sigma)} \cdot  \P[\sigma\text{ is stable}]\]
Moreover, for every permutation $\sigma$ we have:
\begin{itemize}
    \item $\sigma$ and $\sigma^{-1}$ are equally likely to be stable (proved in Lemma~\ref{lemma:inverse}).
    \item $\sigma$ and $\sigma^{-1}$ have the same number of cycles of length > 2 (that is $C(\sigma) = C(\sigma^{-1})$).
    \item $\sigma_{|\M} = \mu_{|\M}$ if and only if $\sigma^{-1}_{|\W} = \mu_{|\W}$
\end{itemize}
Thus, we have $\P[\mu\text{ is stable and women-optimal}] = \P[\mu\text{ is stable and men-optimal}]$.
The matching $\mu$ has the same probability of being the output of MPDA and WPDA, which concludes the proof.
\end{proof}

\subsection{Probability of stability with a non-uniform distribution}

In this subsection, we use the equivalence of Definitions~\ref{def:symmetricantipop} and~\ref{def:symmetricpower} to prove Lemma~\ref{lemma:proba}.

\begin{lemma}
\label{lemma:proba}
Let $P:\M\times\W\rightarrow\mathbb R_{>0}$ define a symmetric anti-popularity preference distribution.
The probability that a fixed permutation $\sigma : \M\cup\W\rightarrow\M\cup\W$ is stable is
\[\P[\sigma \text{ is stable}] = \underbrace{\int\dots\int}_{2N} \mathrm d\mathbf x \cdot \mathrm d \mathbf y \cdot
\hspace{-.2cm}\prod_{\substack{m,w\\\sigma(m) = w\\\sigma(w) \neq m}}\hspace{-.2cm}
\mathbf x_m^{\frac{P(m,\sigma^{-1}(m))}{P(m,w)}}
\hspace{-.2cm}\prod_{\substack{m,w\\\sigma(m) \neq w\\\sigma(w) = m}}\hspace{-.2cm}
\mathbf y_w^{\frac{P(\sigma^{-1}(w),w)}{P(m,w)}}
\hspace{-.2cm}\prod_{\substack{m,w\\\sigma(m) \neq w\\\sigma(w)\neq m}} \hspace{-.2cm} \Big(1- 
\mathbf x_m^{\frac{P(m,\sigma^{-1}(m))}{P(m,w)}}
\mathbf y_w^{\frac{P(\sigma^{-1}(w),w)}{P(m,w)}}\Big)\]
\end{lemma}
\begin{proof}
Using Lemmas~\ref{lemma:symmetric1} and \ref{lemma:symmetric2}, sampling the preference profile from a symmetric anti-popularity preference distribution (with a ``popularity'' function $P$) is equivalent with sampling the preference profile with symmetric power preference distribution (with a ``power'' function $P$). Hence, let $X$ and $Y$ be two random matrices, uniformly sampled from $[0,1]^{\M\times\W}$. Values of $X$ and $Y$ induce a preference profile:
\[\forall m \in \M, \forall w_1, w_2 \in \W,\quad
w_1 \succ_m w_2 \;\Leftrightarrow\; X_{m,w_1}^{P(m,w_1)} < X_{m,w_2}^{P(m,w_2)}\]
\[\forall w \in \W, \forall m_1, m_2 \in \M,\quad
m_1 \succ_w m_2 \;\Leftrightarrow\; Y_{m_1,w}^{P(m_1,w)} < X_{m_2,w}^{P(m_2,w)}\]
As in Subsection~\ref{sec:previousproba} for the uniform case, we condition on the values of $\mathbf x = [X_{m,\sigma^{-1}(m)}]_{m\in\M}$ and $\mathbf y = [Y_{\sigma^{-1}(w),w}]_{w\in\W}$. The permutation $\sigma$ is stable if for all pair $(m,w)$ we have:
\begin{itemize}
    \item If $\sigma(m) = w$ and $\sigma(w) \neq m$, then $w\succ_m\sigma^{-1}(m)$.
    \[\P[w \succ_m \sigma^{-1}(m)] =
    \P\big[X_{m,w}^{P(m,w)} < X_{m,\sigma^{-1}(m)}^{P(m,\sigma^{-1}(m))}\big]
    = \mathbf x_m^{\frac{P(m,\sigma^{-1}(m))}{P(m,w)}}\]
    \item If $\sigma(m) \neq w$ and $\sigma(w) = m$, then  $m \succ_w\sigma^{-1}(w)$.
    \[\P[m \succ_w \sigma^{-1}(w)] =
    \P\big[Y_{m,w}^{P(m,w)} < Y_{\sigma^{-1}(w),w}^{P(\sigma^{-1}(w),w)}\big]
    = \mathbf y_w^{\frac{P(\sigma^{-1}(w),w)}{P(m,w)}}\]
    \item If $\sigma(m) \neq w$ and $\sigma(w) \neq m$, then $\sigma^{-1}(m)\succ_m w$ or $\sigma^{-1}(w) \succ_w m$.
    \begin{align*}
        \P[\sigma^{-1}(m)\succ_m w\text{ or }\sigma^{-1}(w) \succ_w m]
        &= 1 - \P[w \succ_m \sigma^{-1}(m)\text{ and }m \succ_w \sigma^{-1}(w)] \\
        &= 1 - \P\big[X_{m,w}^{P(m,w)} < X_{m,\sigma^{-1}(m)}^{P(m,\sigma^{-1}(m))}\big]
        \cdot \P\big[Y_{m,w}^{P(m,w)} < Y_{\sigma^{-1}(w),w}^{P(\sigma^{-1}(w),w)}\big] \\ 
        &= 1- \mathbf x_m^{\frac{P(m,\sigma^{-1}(m))}{P(m,w)}}
        \mathbf y_w^{\frac{P(\sigma^{-1}(w),w)}{P(m,w)}}
    \end{align*}
\end{itemize}
Conditioning on $\mathbf x$ and $\mathbf y$, each $(m,w)$ property is independent, hence the formula.
\end{proof}

\subsection{Integration by substitution}

In this subsection, we prove that a permutation and its inverse are equally likely to be stable. The intuition is the following. If we build $\sigma$ with the example of Figure~\ref{fig:rotation}, we have a cycle of length 4: $\sigma(m_2) = w_2$, $\sigma(w_2) = m_3$, $\sigma(m_3) = w_3$, $\sigma(w_3) = m_2$. The probability that each person prefer their successor to their predecessor is:
\[\tiny
\underbrace{\frac{P(m_2,w_3)}{P(m_2,w_3)+P(m_2,w_2)}}_{\P[\sigma(m_2) \succ_{m_2} \sigma^{-1}(m_2)]}
\cdot
\underbrace{\frac{P(m_2,w_2)}{P(m_2,w_2)+P(m_3,w_2)}}_{\P[\sigma(w_2) \succ_{w_2} \sigma^{-1}(w_2)]}
\cdot
\underbrace{\frac{P(m_3,w_2)}{P(m_3,w_2)+P(m_3,w_3)}}_{\P[\sigma(m_3) \succ_{m_3} \sigma^{-1}(m_3)]}
\cdot
\underbrace{\frac{P(m_3,w_3)}{P(m_3,w_3)+P(m_2,w_3)}}_{\P[\sigma(w_3) \succ_{w_3} \sigma^{-1}(w_3)]}
\]
However, in the inverse permutation $\sigma^{-1}$ we reverse every edge of the cycle. Observe that the probability that each person prefer their successor to their predecessor remains the same.
\[\tiny
\underbrace{\frac{P(m_2,w_2)}{P(m_2,w_3)+P(m_2,w_2)}}_{\P[\sigma^{-1}(m_2) \succ_{m_2} \sigma(m_2)]}
\cdot
\underbrace{\frac{P(m_3,w_2)}{P(m_2,w_2)+P(m_3,w_2)}}_{\P[\sigma^{-1}(w_2) \succ_{w_2} \sigma(w_2)]}
\cdot
\underbrace{\frac{P(m_3,w_3)}{P(m_3,w_2)+P(m_3,w_3)}}_{\P[\sigma^{-1}(m_3) \succ_{m_3} \sigma(m_3)]}
\cdot
\underbrace{\frac{P(m_2,w_3)}{P(m_3,w_3)+P(m_2,w_3)}}_{\P[\sigma^{-1}(w_3) \succ_{w_3} \sigma(w_3)]}
\]
To incorporate the other conditions of stability, we use the formula from Lemma~\ref{lemma:proba}.
\begin{lemma}
\label{lemma:inverse}
Let $P:\M\times\W\rightarrow\mathbb R_{>0}$ define a symmetric anti-popularity preference distribution.\\
A permutation $\sigma : \M\cup\W\rightarrow\M\cup\W$ and its inverse $\sigma^{-1}$ are equally likely to be stable.
\end{lemma}
\begin{proof}
 From Lemma~\ref{lemma:proba} we have an integral formula for the probability that permutation $\sigma$ is stable. We are going to use an integration by substitution. For each person, we define a function $\varphi$.
 \[
 \forall m \in \M,\quad
 \varphi_m : \left\{\begin{array}{cl}
 [0,1] &\rightarrow\; [0,1] \\
 x &\mapsto\; x^{\frac{P(m,\sigma(m))}{P(m,\sigma^{-1}(m))}}
 \end{array}\right.
 \qquad
 \forall w \in \W,\quad
 \varphi_w : \left\{\begin{array}{cl}
 [0,1] &\rightarrow\; [0,1] \\
 y &\mapsto\; y^{\frac{P(\sigma(w),w)}{P(\sigma^{-1}(w),w)}}
 \end{array}\right.\]
 Each $\varphi$ is a differentiable function with integrable derivative, satisfying $\varphi(0) = 0$ and $\varphi(1) = 1$. When we make the substitution $\int_{\varphi(0)}^{\varphi(1)}  f(t)\mathrm dt = \int_0^1  f(\varphi(t))\varphi'(t)\mathrm dt$, all the terms cancel nicely.
 \[
    \prod_{\substack{m,w\\\sigma(m) \neq w\\\sigma(w)\neq m}} \hspace{-.2cm} \Big(1- 
    \varphi_m(\mathbf x_m)^{\frac{P(m,\sigma^{-1}(m))}{P(m,w)}}
    \varphi_w(\mathbf y_w)^{\frac{P(\sigma^{-1}(w),w)}{P(m,w)}}\Big) =
    \prod_{\substack{m,w\\\sigma^{-1}(m) \neq w\\\sigma^{-1}(w)\neq m}} \hspace{-.2cm} \Big(1- 
    \mathbf x_m^{\frac{P(m,\sigma(m))}{P(m,w)}}
    \mathbf y_w^{\frac{P(\sigma(w),w)}{P(m,w)}}\Big)
\]
\begin{align*}
\prod_m \varphi_m'(\mathbf x_m)
\prod_{\substack{m,w\\\sigma(m) = w\\\sigma(w) \neq m}}\hspace{-.2cm}
\varphi_m(\mathbf x_m)^{\frac{P(m,\sigma^{-1}(m))}{P(m,w)}}
& =
\prod_m {\textstyle\frac{P(m,\sigma(m))}{P(m,\sigma^{-1}(m))}} \cdot 
\mathbf x_m^{\frac{P(m,\sigma(m))}{P(m,\sigma^{-1}(m))} - 1}
\prod_{\substack{m,w\\\sigma(m) = w\\\sigma(w) \neq m}}\hspace{-.2cm} \mathbf x_m
\\ & =
\prod_m {\textstyle\frac{P(m,\sigma(m))}{P(m,\sigma^{-1}(m))} }
\prod_{\substack{m,w\\\sigma(m) = w\\\sigma(w) \neq m}}\hspace{-.2cm}
\mathbf x_m^{\frac{P(m,\sigma(m))}{P(m,\sigma^{-1}(m))}}
\\ & =
\prod_m {\textstyle\frac{P(m,\sigma(m))}{P(m,\sigma^{-1}(m))} }
\prod_{\substack{m,w\\\sigma^{-1}(m) = w\\\sigma^{-1}(w) \neq m}}\hspace{-.2cm}
\mathbf x_m^{\frac{P(m,\sigma(m))}{P(m,w)}}
\end{align*}
Symmetrically, we have
\begin{align*}
\prod_w \varphi_w'(\mathbf y_w)
\prod_{\substack{m,w\\\sigma(m) \neq w\\\sigma(w) = m}}\hspace{-.2cm}
\varphi_w(\mathbf y_w)^{\frac{P(\sigma^{-1}(w),w)}{P(m,w)}}
& =
\prod_w {\textstyle\frac{P(\sigma(w),w)}{P(\sigma^{-1}(w),w)} }
\prod_{\substack{m,w\\\sigma^{-1}(m) \neq w\\\sigma^{-1}(w) = m}}\hspace{-.2cm}
\mathbf y_w^{\frac{P(\sigma(w),w)}{P(m,w)}}
\end{align*}
Finally, the two products of ratios cancel each other:
$\prod_m \frac{P(m,\sigma(m))}{P(m,\sigma^{-1}(m))} 
\prod_w \frac{P(\sigma(w),w)}{P(\sigma^{-1}(w),w)} 
= 1$.
\end{proof}

\subsection{Inclusion-exclusion principle}

In this subsection, we compute the probability that a matching $\mu$ is the men/women-optimal stable matching.
To do so, we use an inclusion-exclusion principle on the set of rotations which could be exposed and women/men-improving in $\mu$.

\begin{lemma}
\label{lemma:optimal}
Let $P:\M\times\W\rightarrow\mathbb R_{>0}$ define a symmetric anti-popularity preference distribution.\\
The probability that a matching $\mu:\M\cup\W\rightarrow\M\cup\W$ is stable and men/women-optimal is
\begin{align*}
\P[\mu\text{ is stable and women-optimal}]
&= \sum_{\substack{\sigma\text{ permutation}\\\sigma_{|\M} = \mu_{|\M}}}
(-1)^{C(\sigma)} \cdot  \P[\sigma\text{ is stable}] \\
\P[\mu\text{ is stable and men-optimal}] &=
\sum_{\substack{\sigma\text{ permutation}\\\sigma_{|\W} = \mu_{|\W}}}
(-1)^{C(\sigma)} \cdot  \P[\sigma\text{ is stable}]
\end{align*}
where $C(\sigma)$ is the number of cycle of length $>2$ in $\sigma$.
\end{lemma}
\begin{proof}
The men and women cases being symmetric, we prove the formula giving the probability that a matching is stable and women-optimal.
Let $\mathcal R$ be the set of rotations $r$ such that $r(m) = \mu(m)$ for all man $m \in r$. Matching $\mu$ is outputted by WPDA when it is stable and women-optimal: no rotation $r \in \mathcal R$ is exposed and women-improving in $\mu$.
\[\P[\mu\text{ is stable and women-optimal}] =
\P[\mu\text{ is stable}] - \P[\mu\text{ is stable and some }r\in \mathcal R\text{ is exposed}]\]
Using an inclusion-exclusion principle to compute the probability of a disjunction, we obtain:
\[\P[\mu\text{ is stable and women-optimal}] =
\sum_{R\subseteq \mathcal R} (-1)^{|R|} \cdot\P[\mu\text{ is stable and every }r\in R\text{ is exposed}]\]
Recall that two different rotations can be exposed at the same time only if they are disjoint. Thus, we can consider only sets $R\subseteq\mathcal R$ of disjoint rotations. Moreover, $\mu$ is stable and every rotation from $R$ is exposed if an only if the associated permutation $\sigma_R$ is stable.
\[\sigma_R : \left\{\begin{array}{ll}
    m \mapsto \mu(m) &\text{if } m\in \M \\
    w \mapsto \mu(w) & \text{if }w \in \W\text{ and }w \notin r \text{ for all } r \in R\\
    w \mapsto r(w) & \text{if }w \in \W\text{ and }w \in r \text{ for some } r \in R\\
\end{array}\right.\]
If $C(\sigma)$ is the number of cycles of length $>2$ in $\sigma$, we have $C(\sigma_R) = |R|$, concluding the proof.
\end{proof}

\section{Future work}

We proved that under certain input distributions, MPDA and WPDA have the same output distribution.
This distribution can be computed by combining Lemmas~\ref{lemma:proba} and~\ref{lemma:optimal}.
In the uniform case, all the matching have the same probability in the output distributions.
Simplifying our formula to find $1/N!$ would be an interesting result.

\medskip
Procedures WPDA and MPDA are two deterministic algorithms which select one stable matching.
It would be interesting to characterize which algorithm also has the same output distribution.
Some candidate mechanisms are studied in \cite{klaus2006procedurally}.
In particular, we believe that our proof extends to the mechanism of
\emph{Employment by Lotto} \cite{aldershof1999refined}, and numerical simulations
suggest it also applies for Roth and Vande Vate's \emph{incremental procedure}
\cite{roth1990random,ma1996randomized,biro2008dynamics}.


\medskip
The ``ex-ante core-converge'' property is a mathematical curiosity, and it does
not imply anything on the strategyproofness of the deferred acceptance algorithms.
However, one economic interpretation is the following.
A decision maker who has prior knowledge on the input distribution of preferences
(\textit{e.g.} from historical data) might try to favor some outcomes
(independently of agents' preferences). We proved that under certain input
distributions, a decision maker who has to chose between the MPDA and WPDA
procedures cannot manipulate (before seeing agents' preferences).

\newpage
\bibliographystyle{ACM-Reference-Format}
\bibliography{biblio}

\end{document}